\documentclass[pre,10pt,twocolumn]{revtex4}%
\usepackage{amsfonts}
\usepackage{amsmath}
\usepackage{amssymb}
\usepackage{graphicx}%
\newtheorem{theorem}{Theorem}

\newtheorem{lemma}[theorem]{Lemma}

\newenvironment{proof}[1][Proof]{\noindent\textbf{#1.} }{\ \rule{0.5em}{0.5em}}
\begin{document}
\title[Quantum Networks on Cubelike Graphs]{Quantum Networks on Cubelike Graphs}
\author{Anna Bernasconi}
\affiliation{Dipartimento di Informatica, Universit\`{a} degli Studi di Pisa, 56127 Pisa, Italy}
\author{Chris Godsil}
\affiliation{Department of Combinatorics \& Optimization, University of Waterloo, Waterloo
N2L 3G1, ON Canada}
\author{Simone Severini}
\affiliation{Institute for Quantum Computing and Department of Combinatorics \&
Optimization, University of Waterloo, Waterloo N2L 3G1, ON Canada}

\begin{abstract}
Cubelike graphs are the Cayley graphs of the elementary abelian group
$\mathbb{Z}_{2}^{n}$ (\emph{e.g.}, the hypercube is a cubelike graph). We
study perfect state transfer between two particles in quantum networks modeled
by a large class of cubelike graphs. This generalizes results of Christandl
\emph{et al. }[\emph{Phys. Rev. Lett.} \textbf{92}, 187902 (2004)] and Facer
\emph{et al. }[\emph{Phys. Rev. A} \textbf{92}, 187902 (2008)].

\end{abstract}
\maketitle

\section{Introduction}

In view of applications like the distribution of cryptographic keys \cite{be,
ek} or the communication between registers in quantum devices \cite{bl, ki},
the study of natural evolution of permanently coupled spin networks has become
increasingly important. A special case of interest consists of homogenous
networks of particles coupled by constant and fixed (nearest-neighbour) interactions.

An important feature of these networks is the possibility of faithfully
transferring a qubit between specific particles without tuning the couplings
or altering the network topology. This phenomenon is usually called
\emph{perfect state transfer} (PST). Since quantum networks (and communication
networks in general) are naturally associated to undirected graphs, there is a
growing amount of literature on the relation between graph-theoretic
properties and properties that allow PST (see \cite{fe, god, j1,ja,sax}).

In the present paper we will give necessary and sufficient conditions for PST
in quantum networks modeled by a large class of \emph{cubelike graphs}. The
vertices of a cubelike graph are the binary $n$-vectors; two vertices $u$ and
$v$ are adjacent if and only if their symmetric difference belongs to a chosen
set. Equivalently, cubelike graphs on $2^{n}$ vertices are the Cayley graphs
of the elementary abelian group of order $2^{n}$ \cite{lo, chu}.

Among cubelike graphs, the hypercube is arguably the most famous one, having
many applications ranging from switching theory to computer architecture,
\emph{etc. }(see, \emph{e.g.}, \cite{le, re1}). There are various and diverse
results about quantum dynamics on hypercubes. These are essentially embraced
by two areas: continuous-time quantum walks \cite{qw}; quantum communication
in spin networks \cite{soug}. The common ingredient is the use of a
Hamiltonian representing the adjacency structure of the graph.

Concerning state transfer, Christandl \emph{et al.} \cite{ch} have shown that
networks modeled by hypercubes are capable of transporting qubits between
pairs of antipodal nodes, perfectly (\emph{i.e.}, with maximum fidelity) and
in constant time. Facer \emph{et al. }\cite{fa} generalized this observation,
by considering a family of cubelike graphs whose members have the hypercube as
a spanning subgraph (for this reason, these authors coined the term
\emph{dressed hypercubes}). Other questions related to quantum dynamics on
hypercubes have been addressed in \cite{al, chi, kro, lo1, moo}.

The paper is organized as follows: Section 2 contains the necessary
definitions and the statements of our results. A proof will be given in
Section 3. This is obtained by diagonalizing the Hamiltonians with simple
tools from Fourier analysis on $\mathbb{Z}_{2}^{n}$.

\section{Set up and results}

Let $\mathbb{Z}_{2}^{n}$ be the additive abelian group $(\mathbb{Z}%
_{2})^{\times n}$. Each element of $\mathbb{Z}_{2}^{n}$ is represented as a
binary vector of length $n$. The \emph{zero vector} $\mathbf{0}$ is made up of
all $0$'s. Let $f:\mathbb{Z}_{2}^{n}\rightarrow\mathbb{Z}_{2}$ be a Boolean
function on $n$ variables and let $\Omega_{f}=\{w\in\mathbb{Z}_{2}%
^{n}\ |\ f(w)=1\}$. Let $d=|\Omega_{f}|$ be the number of vectors
$w\in\mathbb{Z}_{2}^{n}$ such that $f(w)=1$. Finally, if $w$ and $v$ are two
binary vectors of the same length, then $w\oplus v$ denotes the vector
obtained by computing their elementwise addition modulo $2$, and $w^{T}v$
their scalar product.

The \emph{Cayley graph} $X(\Gamma,T)$ of a group $\Gamma$ \emph{w.r.t.} the
set $T\subseteq\Gamma$ ($T=T^{-1}$) is the graph with vertex-set
$V(X)=\{\Gamma\}$ and an edge $\{g,h\}\in E(X)$, if there is $s\in T$ such
that $gs=h$. The set $T$ is also called \emph{Cayley set}. The Cayley graphs
of the form $X(\mathbb{Z}_{2}^{n},\Omega_{f})$ are called \emph{cubelike
graphs}. Some cubelike graphs are illustrated in Figure 1 below.%

\begin{figure}
[h]
\begin{center}
\includegraphics[
height=1.2341in,
width=3.4904in
]%
{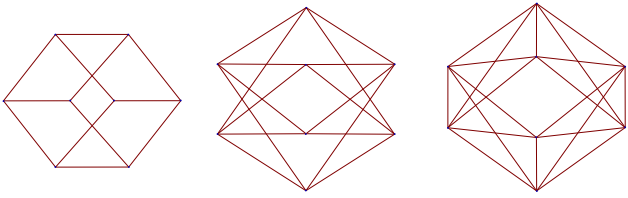}%
\caption{Drawings of three nonisomorphic cubelike graphs on $8$ vertices. }%
\end{center}
\end{figure}

Notice that this definition embraces every possible set $\Omega_{f}$. When $f$
is the characteristic function of the standard generating set of
$\mathbb{Z}_{2}^{n}$, the graph $X(\mathbb{Z}_{2}^{n},\Omega_{f})$ is called
\emph{hypercube}. (For example, the leftmost graph in Figure 1). The
\emph{adjacency matrix} of $X(\mathbb{Z}_{2}^{n},\Omega_{f})$ is the
$2^{n}\times2^{n}$ matrix
\[
A_{f}=\sum_{w\in\Omega_{f}}\rho_{reg}(w),
\]
where $\rho_{reg}(x)$ is the regular (permutation) representation of
$w\in\Omega_{f}$. In particular, if $w=w_{1}w_{2}\cdots w_{n}$ then%
\[
\rho_{reg}(w)=\bigotimes_{i=1}^{n}\sigma_{x}^{w_{i}},
\]
where $\sigma_{x}$ is a Pauli matrix. It is clear that $A_{f}$ commutes with
the adjacency matrix of any other cubelike graph, given that the group
$\mathbb{Z}_{2}^{n}$ is abelian.

Now, let us choose a bijection between vertices of $X(\mathbb{Z}_{2}%
^{n},\Omega_{f})$ and the elements of the standard basis $|1\rangle
,|2\rangle,...,|N\rangle$ of an Hilbert space $\mathcal{H}\cong\mathbb{C}^{N}%
$, where $N=2^{n}$. This is the usual space of $n$ qubits. If we look at the
single excitation case in the XY model, the evolution of a network of spin
$1/2$ quantum mechanical particles on the vertices of $X(\mathbb{Z}_{2}%
^{n},\Omega_{f})$ can be seen as induced by the adjacency matrix $A_{f}$,
which then plays the role of an Hamiltonian (for details see \cite{fa} or
\cite{ch}). On the light of this observation, given two vectors $a,b\in
\mathbb{Z}_{2}^{n}$, the (unnormalized) \emph{transition amplitude} between
$a$ and $b$ induced by $A_{f}$ is the expression
\begin{align}
T(a,b)  &  =\langle b|e^{-iA_{f}t}|a\rangle\\
&  =\sum_{w\in\mathbb{Z}_{2}^{n}}(-1)^{a^{T}w}e^{-i\lambda_{w}t}(-1)^{b^{T}%
w}\nonumber\\
&  =\sum_{w\in\mathbb{Z}_{2}^{n}}(-1)^{(a\oplus b)^{T}w}e^{-i\lambda_{w}%
t},\nonumber
\end{align}
where $t\in\mathbb{R}^{+}$. The \emph{fidelity} of state transfer between $a$
and $b$ is then $F(a,b)=\frac{1}{2^{n}}\left\vert T(a,b)\right\vert $. By
definition, the evolution under $A_{f}$ is \emph{periodic} if there is
$t\in\mathbb{R}^{+}$ such that $F(a,a)=1$ for every $a\in\mathbb{Z}_{2}^{n}$.
Graph-theoretic properties responsible of periodic evolution (sometime also
called \emph{perfect revival}) have been considered in the literature
(\cite{ge, god, mu, sax, vo}).

With the next proposition, we show that every network modeled by a cubelike
graph has a periodic evolution. The period is $\pi$ and it does not depend on
the number of vertices of the graph. Equivalently, it does not depend on the
dimension of $\mathcal{H}$.

\begin{theorem}
\label{the}Let $X(\mathbb{Z}_{2}^{n},\Omega_{f})$ be a cubelike graph and let
$a,b\in\mathbb{Z}_{2}^{n}$.

\begin{enumerate}
\item For $t=\pi$, we have $F(a,b)=1$ if and only if $a=b$.

\item For $t=\pi/2$, we have $F(a,b)=1$ if $a\oplus b=u$ and $u=\bigoplus
_{w\in\Omega_{f}}w\neq\mathbf{0}$.
\end{enumerate}
\end{theorem}

As a simple consequence of this statement, we have various ways to
\emph{route} information between any two nodes of a network whose vertices
correspond to the elements of $\mathbb{Z}_{2}^{n}$. Let $f$ be a Boolean
function such that $\Omega_{f}=\{w_{1},...,w_{r}\}$ is a generating set of
$\mathbb{Z}_{2}^{n}$. Let $\bigoplus_{w_{i}\in\Omega_{f}}w_{i} = w
\neq\mathbf{0}$. Let us define $C=\{w,w_{1},...,w_{r}\}$ and $C_{i}%
=C\backslash w_{i}$. Since the sum of the elements of $C_{i}$ is nonzero, the
Cayley graph $X(\mathbb{Z}_{2}^{n},C_{i})$ has PST between $a$ and $b$ such
that $a\oplus b=w_{i}$ at time $\pi/2$. The simplest case arises when we chose
$r=n$ and take the vectors $w_{1},...,w_{r}$ to be the standard basis of
$\mathbb{Z}_{2}^{n}$. Then $X(\mathbb{Z}_{2}^{n},C)$ is the \emph{folded }%
$d$\emph{-cube} and, by using a suitable sequence of the graphs $X(\mathbb{Z}%
_{2}^{n},C_{i})$, we can arrange PST from the zero vector to any desired
element of $\mathbb{Z}_{2}^{n}$.

For example, consider the case $d=3$. We write $w_{1}=(100)$, $w_{2}=(010)$
and $w_{3}=(001)$. Then $w=(111)$. The graph $X(\mathbb{Z}_{2}^{3},C)$ is
illustrated in Figure 1 -- left. Since $w_{1}=w_{2}\oplus w_{3}\oplus w$,
there is PST between $000$ and $w_{1}=100$ at time $\pi/2$, given that
$000\oplus100=100$ (see Figure 2 -- right). Also, there is PST for the pairs
$\{010,110\}$, $\{001,101\}$ and $\{011,111\}$. Notice that $X(\mathbb{Z}%
_{2}^{n},C_{1})$ is isomorphic to the $3$-dimensional hypercube.%

\begin{figure}
[h]
\begin{center}
\includegraphics[
height=1.8795in,
width=3.4406in
]%
{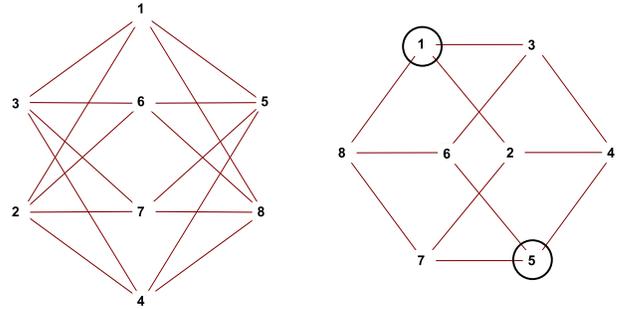}%
\caption{Left: The graph $X(\mathbb{Z}_{2}^{n},C)$, where
$C=\{(100),(010),(001),(111)\}$. This graph is isomorphic to the complete
bipartite graph $K_{4,4}$. Right: The graph $X(\mathbb{Z}_{2}^{n},C_{1})$,
where $C_{1}=\{(010),(001),(111)\}$. This graph turns out to be isomorphic to
the hypercube of dimension $3$. Since $100=010\oplus001\oplus111$, there is
PST between $000$ and $100$ at time $\pi/2$, given that $000\oplus100=100$.
Also, there is PST for the pairs $\{010,110\}$, $\{001,101\}$ and
$\{011,111\}$.}%
\end{center}
\end{figure}

For generic dimension, the hypercube is $X(\mathbb{Z}_{2}^{n},E)$, where
$E=\{(10...0),(010...0),...,(0...01)\}$ and $\overrightarrow{J_{n}}%
=\bigoplus_{w\in E}w$, the all-ones vector of length $2^{n}$. Since the
Hamming distance between two different elements $a,b\in\mathbb{Z}_{2}^{n}$ is
exactly their distance in $X(\mathbb{Z}_{2}^{n},E)$, we have PST between any
two antipodal vertices of the hypercube, as it was already observed in
\cite{ch} and \cite{fa}. Recall that the distance between two vertices in a
graph is the length of the geodesic (equivalently, the shortest path)
connecting the vertices. Two vertices are said to be \emph{antipodal} if their
distance is the \emph{diameter} of the graph, \emph{i.e.}, the longest among
all the geodesics. Antipodal vertices in a Cayley graphs are connected via a
sequence of all elements of the Cayley set. It remains as an open problem to
verify that whenever there is PST between two vertices of a cubelike graph
then the vertices are antipodal.

\section{Proof of the theorem}

The \emph{abstract Fourier transform} of a Boolean function $f$ is the
rational valued function $f^{\ast}:\mathbb{Z}_{2}^{n}\rightarrow\mathbb{Q}$
which defines the coefficients of $f$ with respect to the orthonormal basis of
the functions%
\[
Q_{w}(x)=(-1)^{w^{T}x},
\]
that is,
\[
f^{\ast}(w)=2^{-n}\sum_{x\in\mathbb{Z}_{2}^{n}}(-1)^{w^{T}x}f(x).
\]
Then
\[
f(x)=\sum_{w\in\mathbb{Z}_{2}^{n}}(-1)^{w^{T}x}f^{\ast}(w)
\]
is the Fourier expansion of $f$. Note that the zero-order Fourier coefficient
is equal to the probability that the function takes the value $1$,
\emph{i.e.}, $f^{\ast}(\mathbf{0)}=\frac{d}{2^{n}}$, while the other Fourier
coefficients measure the correlation between the function and the parity of
subsets of its arguments.

Using a vector-representation for the functions $f$ and $f^{\ast}$, and
considering the natural ordering of the binary vectors $w\in\mathbb{Z}_{2}%
^{n}$, one can derive a convenient matrix formulation for the transform pair:
$f=H_{n}\,f^{\ast}$ and $f^{\ast}=\frac{1}{2^{n}}H_{n}\,f$, where $H_{n}$ is
the Hadamard transform matrix. Given a function $f:\mathbb{Z}_{2}%
^{n}\rightarrow\mathbb{Z}_{2}$, the set $\Omega_{f}$ defines the Cayley graph
$X(\mathbb{Z}_{2}^{n},\Omega_{f})$, whose spectrum coincides, up to a factor
$2^{n}$, with the Fourier spectrum of the function: $\frac{1}{2^{n}}H_{n}%
A_{f}H_{n}=D_{f}$, where $A_{f}$ is the adjacency matrix of $X(\mathbb{Z}%
_{2}^{n},\Omega_{f})$ and $D_{f}$ is an $2^{n}\times2^{n}$ diagonal matrix. In
particular, for $w\in\mathbb{Z}_{2}^{n}$, we have $\lambda_{w}=2^{n}f^{\ast
}(w)$. Theorem \ref{the} needs the following two technical lemmas.

\begin{lemma}
\label{Somma0} Let $f$ be a Boolean function such that $\bigoplus_{w\in
\Omega_{f}}w=\mathbf{0}$. Then, for all $v\in\mathbb{Z}_{2}^{n}$,
\[
\lambda_{v}=d-4k_{v},
\]
where $d=|\Omega_{f}|$ and $k_{v}\in\mathbb{N}$, $0\leq k_{v}\leq\lfloor
d/2\rfloor$.
\end{lemma}

\begin{proof}
Let $a\in\Omega_{f}$, $a\neq\mathbf{0}$, and let $\Omega_{f}^{\prime}%
=\Omega\setminus\{a\}$. Since $\bigoplus_{w\in\Omega_{f}}w=\mathbf{0}$, we
have that $a=\bigoplus_{w\in\Omega_{f}^{\prime}}w$. Now observe that, for any
$v\in\mathbb{Z}_{2}^{n}$,
\begin{align*}
\lambda_{v}  &  =\sum_{w\in\mathbb{Z}_{2}^{n}}(-1)^{w^{T}v}f(w)=\sum
_{w\in\Omega_{f}}(-1)^{w^{T}v}\\
&  =\sum_{w\in\Omega_{f}^{\prime}}(-1)^{w^{T}v}+(-1)^{a^{T}v}\\
&  =\sum_{w\in\Omega_{f}^{\prime}}(-1)^{w^{T}v}+(-1)^{(\bigoplus_{w\in
\Omega_{f}^{\prime}}w)^{T}v}\\
&  =\sum_{w\in\Omega_{f}^{\prime}}(-1)^{w^{T}v}+\prod_{w\in\Omega_{f}^{\prime
}}(-1)^{w^{T}v}\,.
\end{align*}
Let $n_{v}=|\{w\in\Omega^{\prime}\ |\ (-1)^{w^{T}v}=-1\}|$, and $p_{v}%
=|\{w\in\Omega^{\prime}\ |\ (-1)^{w^{T}v}=1\}|$. Note that $p_{v}=|\Omega
_{f}^{\prime}|-n_{v}=(d-1)-n_{v}$. We have
\[
\lambda_{v}=p_{v}-n_{v}+(-1)^{n_{v}}=d-1-2n_{v}+(-1)^{n_{v}}.
\]
Thus, if $n_{v}$ is an even number, $n_{v}=2k_{v}$, with $k_{v}\in\mathbb{N}$,
we get
\[
\lambda_{v}=d-1-4k_{v}+1=d-4k_{v}.
\]
Otherwise, if $n_{v}=2k_{v}-1$ is odd, we get
\[
\lambda_{v}=d-1-2(2k_{v}-1)-1=d-4k_{v}.
\]
Since, for all $v$, $-d \leq\lambda_{v}\leq d$, we finally get $0 \leq
k_{v}\leq\lfloor d/2\rfloor$.
\end{proof}

\begin{lemma}
\label{SommaNon0} Let $f$ be a Boolean function such that $\bigoplus
_{w\in\Omega_{f}}w=u\neq\mathbf{0}$.

\begin{enumerate}
\item If $u\not \in \Omega_{f}$, then, for all $v\in\mathbb{Z}_{2}^{n}$,
\[
\lambda_{v}=\left\{
\begin{array}
[c]{ll}%
d-4k_{v}, & \mbox{\em if }u^{T}v\mbox{\em \ is even};\\
d-4k_{v}+2, & \mbox{\em if }u^{T}v\mbox{\em \ is odd};
\end{array}
\right.
\]
where $d=|\Omega_{f}|$ and $k_{v}\in\mathbb{N}$, $0 \le k_{v} \le\lfloor
\frac{d+1}{2} \rfloor$.

\item If $u\in\Omega_{f}$, then, for all $v\in\mathbb{Z}_{2}^{n}$,
\[
\lambda_{v}=\left\{
\begin{array}
[c]{ll}%
d-4k_{v}, & \mbox{\em if }u^{T}v\mbox{\em \ is even};\\
d-4k_{v}-2, & \mbox{\em if }u^{T}v\mbox{\em \ is odd};
\end{array}
\right.
\]
where $d=|\Omega_{f}|$ and $k_{v}\in\mathbb{N}$, $0 \le k_{v} \le\lfloor
\frac{d-1}{2} \rfloor$.
\end{enumerate}
\end{lemma}

\begin{proof}
\noindent\emph{(1.) }Consider the function $g$ such that $\Omega_{g}%
=\Omega_{f}\cup\{u\}$. Let $\mu_{v}$ denote the eigenvalues of the Cayley
graph associated to $g$. As $\bigoplus_{w\in\Omega_{g}}w=\mathbf{0}$, from
Lemma \ref{Somma0}, we get
\[
\mu_{v}=|\Omega_{g}|-4k_{v}=d+1-4k_{v},
\]
with $0 \le k_{v} \le\lfloor\frac{d+1}{2} \rfloor$. Now, observe that
\begin{align*}
\mu_{v}  &  =\sum_{w\in\mathbb{Z}_{2}^{n}}(-1)^{w^{T}v}g(w)=\sum_{w\in
\Omega_{g}}(-1)^{w^{T}v}\\
&  =\sum_{w\in\Omega_{f}}(-1)^{w^{T}v}+(-1)^{u^{T}v}=\lambda_{v}+(-1)^{u^{T}%
v}.
\end{align*}
Thus,%
\[
\lambda_{v}=\mu_{v}-(-1)^{u^{T}v}=d+1-4k_{v}-(-1)^{u^{T}v},
\]
and the thesis immediately follows.

\medskip\noindent\emph{(2.) }Consider the function $g$ such that $\Omega
_{g}=\Omega_{f}\setminus\{u\}$. Let $\mu_{v}$ denote the eigenvalues of the
Cayley graph associated to $g$. As $u=\bigoplus_{w\in\Omega_{f}}%
w=\bigoplus_{w\in\Omega_{g}}\oplus u$, we have $\bigoplus_{w\in\Omega_{g}%
}w=\mathbf{0}$. By applying Lemma \ref{Somma0}, we obtain
\[
\mu_{v}=|\Omega_{g}|-4k_{v}=d-1-4k_{v},
\]
with $0 \le k_{v} \le\lfloor\frac{d-1}{2} \rfloor$. Now, as in \emph{(1.)}
observe that
\begin{align*}
\lambda_{v}  &  =\sum_{w\in\mathbb{Z}_{2}^{n}}(-1)^{w^{T}v}f(w)=\sum
_{w\in\Omega_{f}}(-1)^{w^{T}v}\\
&  =\sum_{w\in\Omega_{g}}(-1)^{w^{T}v}+(-1)^{u^{T}v}=\mu_{v}+(-1)^{u^{T}v}.
\end{align*}
So we get%
\[
\lambda_{v}=d-1-4k_{v}+(-1)^{u^{T}v},
\]
concluding the proof of the lemma.
\end{proof}

\bigskip

\begin{proof}
[Proof of Theorem 1]\emph{(1.) }All eigenvalues $\lambda_{w}$ are integers
with the same parity. In particular, they are all odd if $d=|\Omega_{f}|$ is
odd, and all even, otherwise. Thus, by Eq. (1), for $t=\pi$ we have%
\begin{align*}
T(a,b)  &  =\sum_{w\in\mathbb{Z}_{2}^{n}}(-1)^{(a\oplus b)^{T}w}%
e^{-i\lambda_{w}\pi}\\
&  =\sum_{w\in\mathbb{Z}_{2}^{n}}(-1)^{(a\oplus b)^{T}w}\cos(\lambda_{w}\pi)\\
&  =\sum_{w\in\mathbb{Z}_{2}^{n}}(-1)^{(a\oplus b)^{T}w}(-1)^{\lambda_{w}}\\
&  =(-1)^{d}\sum_{w\in\mathbb{Z}_{2}^{n}}(-1)^{(a\oplus b)^{T}w}.
\end{align*}
Thus the statement follows since
\[
\sum_{w\in\mathbb{Z}_{2}^{n}}(-1)^{(a\oplus b)^{T}w}=\left\{
\begin{tabular}
[c]{ll}%
$2^{n},$ & $a\oplus b=\mathbf{0}$, \emph{i.e.}, $a=b$;\\
$0,$ & otherwise.
\end{tabular}
\ \ \ \ \ \ \right.
\]
\emph{(2.) }First, let us suppose that $f$ is such that $\bigoplus_{w\in
\Omega_{f}}w=\mathbf{0}$. Then, applying Lemma \ref{Somma0}, we can see that
\begin{align*}
T(a,b)  &  =\sum_{w\in\mathbb{Z}_{2}^{n}}(-1)^{(a\oplus b)^{T}w}%
e^{-i\lambda_{w}\pi/2}\\
&  =\sum_{w\in\mathbb{Z}_{2}^{n}}(-1)^{(a\oplus b)^{T}w}e^{-i(d-4k_{w})\pi
/2}\\
&  =e^{-id\pi/2}\sum_{w\in\mathbb{Z}_{2}^{n}}(-1)^{(a\oplus b)^{T}w}%
e^{i2k_{w}\pi}\\
&  =e^{-id\pi/2}\sum_{w\in\mathbb{Z}_{2}^{n}}(-1)^{(a\oplus b)^{T}w}.
\end{align*}
The thesis is verified, since $\sum_{w\in\mathbb{Z}_{2}^{n}}(-1)^{(a\oplus
b)^{T}w}=2^{n}$ if and only if $a\oplus b=\mathbf{0}=\bigoplus_{w\in\Omega
_{f}}w$, but here $a\neq b$. The remaining case is when $\bigoplus_{w\in
\Omega_{f}}w=u\neq\mathbf{0}$. Applying Lemma \ref{SommaNon0}, we can write
\begin{align*}
T(a,b)  &  =\sum_{%
\genfrac{}{}{0pt}{}{w\in\mathbb{Z}_{2}^{n}}{u^{T}w\mbox{\tiny \  even}}%
}(-1)^{(a\oplus b)^{T}w}e^{-i(d-4k_{w})\pi/2}\\
&  +\sum_{%
\genfrac{}{}{0pt}{}{w\in\mathbb{Z}_{2}^{n}}{u^{T}w\mbox{\tiny \  odd}}%
}(-1)^{(a\oplus b)^{T}w}e^{-i(d-4k_{w}\pm2)\pi/2}\\
&  =e^{-id\pi/2}\left(  \sum_{%
\genfrac{}{}{0pt}{}{w\in\mathbb{Z}_{2}^{n}}{u^{T}w\mbox{\tiny \ even}}%
}(-1)^{(a\oplus b)^{T}w}e^{i 2k_{w}\pi}\right. \\
&  \left.  -\sum_{%
\genfrac{}{}{0pt}{}{w\in\mathbb{Z}_{2}^{n}}{u^{T}w\mbox{\tiny \  odd}}%
}(-1)^{(a\oplus b)^{T}w}e^{i2k_{w}\pi}\right) \\
&  =e^{-id\pi/2}\sum_{w\in\mathbb{Z}_{2}^{n}}(-1)^{(a\oplus b)^{T}%
w}(-1)^{u^{T}w}\\
&  =e^{-id\pi/2}\sum_{w\in\mathbb{Z}_{2}^{n}}(-1)^{(a\oplus b\oplus u)^{T}w}.
\end{align*}
The statement holds since $\sum_{w\in\mathbb{Z}_{2}^{n}}(-1)^{(a\oplus b\oplus
u)^{T}w}=2^{n}$ if and only if $a\oplus b\oplus u=\mathbf{0}$, \emph{i.e.},
$a\oplus b=u=\bigoplus_{w\in\Omega_{f}}w$.
\end{proof}

\section{Conclusion}

We have given a necessary and sufficient condition for PST in quantum networks
modeled by a large class of cubelike graphs.

A special case is left open: when $\bigoplus_{w\in\Omega_{f}}w=\mathbf{0}$.
Numerical evidence suggests that cubelike graphs with this property do not
allow PST.

An application of our result is a straightforward method to distinguish if,
for a Boolean function $f$, we have $\bigoplus_{w\in\Omega_{f}}w\neq
\mathbf{0}$ or $\bigoplus_{w\in\Omega_{f}}w=\mathbf{0}$, when this promise
holds. By Theorem \ref{the}, we know that $\langle b|e^{-iA_{f}\pi/2}%
|a\rangle=1$ if $a\oplus b=w\neq\mathbf{0}$. If we let evolve the system under
the Hamiltonian $A_{f}$ for a time $\pi/2$, we then obtain the state
$|b\rangle$ by performing a von Newmann measurement \emph{w.r.t} the standard
basis of $\mathcal{H}$. When $n$ gets large, the probability of observing
$|b\rangle$ with distinct measurements tends to decrease exponentially if
$\bigoplus_{w\in\Omega_{f}}w=\mathbf{0}$.

On the other side, this problem is easy without the use of any quantum
technique. Beyond this trivial application, given the link between cubelike
graphs and Boolean functions, it is natural to ask weather quantum dynamics on
these graphs can help in getting useful information about the corresponding functions.

\bigskip

\noindent\emph{Acknowledgments. }We thank Sougato Bose, Andrew Childs and
Peter H\o yer for helpful discussion. Part of this work was carried on during
the conference \textquotedblleft Graph Theory and Quantum
Information:\ Emerging Connections\textquotedblright, held at Perimeter
Institute for Theoretical Physics (April 28 - May 2, 2008, Waterloo).

\end{document}